\documentclass{article}
\usepackage{fullpage}
\usepackage{amsfonts, amsmath, amssymb, amsthm} 
\usepackage{authblk}

\usepackage{graphicx}
\usepackage{tikz} 
\usetikzlibrary{calc,shapes,arrows.meta}
\usepackage{mathabx}
\usepackage{subcaption}
\usepackage{todonotes}
\usepackage{hyperref}
\usepackage[capitalise]{cleveref}
\usepackage{enumerate}
\usepackage{multirow}
\usepackage{verbatim}
\usepackage[ruled,noend,noline]{algorithm2e}
\DontPrintSemicolon
\newcommand{\LeftComment}[1]{$\triangleright$ #1}
\usepackage[noend]{algpseudocode}
\newcommand{\cO}{\mathcal{O}}

\def\dd{\mathinner{.\,.}}

\renewcommand{\Alph}{\mathit{Alph}}
\newcommand{\first}{\mathit{first}}
\newcommand{\last}{\mathit{last}}
\newcommand{\PRED}{\mathit{PRED}\xspace}
\newcommand{\Pref}{\mathit{Pref}\xspace}
\newcommand{\FirstLex}{\mathit{FirstOcc}\xspace}
\newcommand{\LastLex}{\mathit{LastOcc}\xspace}
\newcommand{\shrink}{\mathit{shrink}\xspace}
\newcommand{\falafel}{\mathsf{FaLaFeL}\xspace}
\newcommand{\D}{FL}%{\mathbf{D}}

\newtheorem{theorem}{Theorem}[section]
\newtheorem{fact}[theorem]{Fact}
\newtheorem{lemma}[theorem]{Lemma}
\newtheorem{observation}[theorem]{Observation}

\newtheorem{corollary}[theorem]{Corollary}
\newtheorem{claim}[theorem]{Claim}

\theoremstyle{remark}
\newtheorem{definition}[theorem]{Definition}
\newtheorem{remark}[theorem]{Remark}
\newtheorem{example}[theorem]{Example}

\author[1]{Panagiotis Charalampopoulos}
\author[2,3]{Solon P.\ Pissis}
\author[4]{Jakub Radoszewski}
\author[4]{Wojciech Rytter}
\author[4]{Tomasz Wale\'n}
\author[2,4]{Wiktor Zuba}

\affil[1]{King's College London, UK\\\texttt{p.charalampopoulos@kcl.ac.uk}}
\affil[2]{CWI, Amsterdam, The Netherlands\\\texttt{solon.pissis@cwi.nl}}
\affil[3]{Vrije Universiteit, Amsterdam, The Netherlands}
\affil[4]{Institute of Informatics, University of Warsaw, Warsaw, Poland\\\texttt{[jrad,rytter,walen,w.zuba]@mimuw.edu.pl}}

\title{Subsequence Covers of Words}

\date{\vspace{-.5cm}}

\begin{document}

\maketitle

\begin{abstract}
We introduce subsequence covers (s-covers, in short), a new type of covers of a word. A word $C$ is an \emph{s-cover} of a word $S$ if the occurrences of $C$ in $S$ as subsequences cover all the positions in $S$.

The s-covers seem to be computationally much harder than standard covers of words (cf.\ Apostolico et al., \emph{Inf.\ Process.\ Lett.} 1991), but, on the other hand, much easier than the related shuffle powers (Warmuth and Haussler, \emph{J.\ Comput.\ Syst.\ Sci.} 1984).

We give a linear-time algorithm for testing if a candidate word $C$ is an s-cover of a word $S$ over a polynomially-bounded integer alphabet. We also give an algorithm for finding a shortest s-cover of a word $S$, which in the case of a constant-sized alphabet, also runs in linear time. 

The words without proper s-cover are called s-primitive.
We complement our algorithmic results with explicit  lower and an upper bound on the length of a longest s-primitive word.  Both bounds are exponential in the size of the alphabet. The upper bound presented here 
improves the bound given in the conference version of this paper [SPIRE 2022].

\medskip\noindent\textbf{Keywords:} Text algorithms, Combinatorics on words,  Covers, Shuffle powers, Subsequence
\end{abstract}

\section{Introduction} 

The problem of computing covers in a word is a classic one in string algorithms; see 
\cite{DBLP:journals/ipl/ApostolicoFI91,DBLP:journals/ipl/Breslauer92,DBLP:journals/ipl/MooreS95} and also~\cite{DBLP:journals/tcs/CzajkaR21,DBLP:journals/fuin/MhaskarS22} for recent surveys. In its most basic type, a word $C$ is said to be a \emph{cover} of a word $S$ if every position of $S$ lies within some occurrence of $C$ as a factor (subword) in $S$~\cite{DBLP:journals/ipl/ApostolicoFI91}.

In this paper, we introduce a new type of cover, in which instead of factors we take subsequences (scattered factors). Such covers turn out to be related to shuffle problems~\cite{DBLP:journals/jcss/WarmuthH84,DBLP:conf/csr/RizziV13,DBLP:journals/jcss/BussS14}.
Formally the new type of cover is defined as follows:

\begin{definition}
A word $C$ is a {\bf subsequence cover} (s-cover, in short) of a word $S$ if every position in $S$ belongs to an occurrence of $C$ as a subsequence in $S$. We also write $S\in C^{\,\otimes}$, where $C^{\,\otimes}$
is the set of words having $C$ as an s-cover.
\end{definition}

We say that an s-cover $C$ of a word $S$ is \emph{non-trivial} if $|C|<|S|$. A word $S$ is called \emph{s-primitive} if it has no non-trivial s-cover.

\begin{remark}
An example of an s-primitive word is the 
Zimin word $Z_k$~\cite{lothaire_2002}, that is, a word over alphabet $\{1,\ldots,k\}$ given by recurrences 
\begin{equation}\label{zimin}
Z_1=1,\ \ Z_i=Z_{i-1}iZ_{i-1}\  \text{for}\ i \in \{2,\ldots,k\}.
\end{equation}
The word $Z_k$ has length $2^k-1$. The fact that it is s-primitive can be shown by simple induction.
\end{remark}

Clearly, if a word $C$ is a (standard) cover of a word $S$, then $C$ is an s-cover of~$S$.
However, the converse implication is false: $ab$ is an s-cover of $aab$, but is not a standard cover.
For another example of an s-cover, see below.

\begin{example}
\cref{ex:simple-s-cover} shows that $C=abcab$ is an s-cover of $S=abcbacab$. In fact $C$ is a shortest s-cover of $S$.
\begin{figure}[ht]
\centering
\begin{tikzpicture}[xscale=0.35]
    
    \foreach \i/\c in {1/a,2/b,3/c,5/a,8/b}{
        \draw (\i,1.5) node[above] {$\textcolor{violet}{\c}$};
    }
    \foreach \i/\c in {1/a,4/b,6/c,7/a,8/b}{
        \draw (\i,1) node[above] {$\textcolor{red}{\c}$};
    }
    \foreach \i/\c in {1/a,2/b,3/c,7/a,8/b}{
        \draw (\i,0.5) node[above] {$\textcolor{blue}{\c}$};
    }
    \foreach \i/\c in {1/a,2/b,3/c,4/b,5/a,6/c,7/a,8/b}{
        \draw (\i,0) node[above] {$\textcolor{black}{\bold \c}$};
    }
\end{tikzpicture}
    \caption{An illustration of the fact that $C=abcab$ is an s-cover of $S=abcbacab$.}
    \label{ex:simple-s-cover}
\end{figure}
\end{example}

We now provide some basic definitions and notation.
An \emph{alphabet} is a finite nonempty set of elements called \emph{letters}.
A \emph{word} $S$ is a sequence of letters over some alphabet.
For a word $S$, by $|S|$ we denote its \emph{length}, by $S[i]$, for $i=0,\ldots,|S|-1$, we denote its $i$th letter, and by $\Alph(S)$ we denote the set of letters in $S$, i.e., $\{S[0],\ldots,S[|S|-1]\}$.
The \emph{empty word} is the word of length 0.

For any two words $U$ and $V$, by $U\cdot V = UV$ we denote their concatenation. For a word $S=PUQ$, where $P$, $U$, and $Q$ are words, $U$ is called a \emph{factor} of $S$; it is called a \emph{prefix} (resp.~\emph{suffix}) if $P$ (resp.~$Q$) is the empty word.
By $S[i\dd j]=S[i \dd j+1)=S(i-1 \dd j]$ we denote a factor $S[i] \dots S[j]$ of $S$; we omit $i$ if $i=0$ and $j$ if $j=|S|-1$.

A word $V$ is a \emph{$k$-power} of a word $U$, for integer $k \ge 0$, if $V$ is a concatenation of $k$ copies of $U$, in which case we denote it by $U^k$. It is called a \emph{square} if $k=2$.

\begin{remark}\label{SLP} If a word contains a factor which is not s-primitive, 
then the whole word is not s-primitive. It is easy to see that any gapped repeat $UVU$ (see \cite{DBLP:conf/cpm/KolpakovPPK14}), where $\Alph(V) \subseteq \Alph(U)$, is not s-primitive and $U$ is its non-trivial s-cover. This is the case because every position of $V$ can be covered by an occurrence whose prefix belongs to the first copy of $U$ and its suffix to the second copy of $U$ with only the one position of $V$ used.
In particular nontrivial squares $UU$ are of this type.
Each word containing a factor of this type is not s-primitive.
% WZ: This is so long because we added the argument asked by Reviewer 1 and the majority of authors agrees with adding this change.
\begin{comment}
If a word $S$ contains a non-empty square factor $U^2$, then $S$ has a non-trivial s-cover resulting by removing any of the two consecutive copies of $U$. Further, if a word $S$ has a factor being a gapped repeat $UVU$ (see \cite{DBLP:conf/cpm/KolpakovPPK14}), such that $\Alph(V) \subseteq \Alph(U)$, then $S$ has a non-trivial s-cover obtained by deleting from $S$ the suffix $VU$ of this gapped repeat. %\textcolor{red}
{This is the case because $U$ is an s-cover of $UVU$ (every position of $V$ can be covered by an occurrence whose prefix belongs to the first copy of $U$ and its suffix to the second copy of $U$ with only the one position of $V$ used).} 
Moreover, if $C$ is an s-cover of $S$, then $C$ is an s-cover of $S$ concatenated with any concatenation of suffixes of $C$.
\end{comment}
\end{remark}

A different version of covers, where we require that position-subsequences\footnote{
%\textcolor{red}
{By a \emph{position-subsequence} we mean the sorted set of positions that witness the occurrence of the subsequence.}} are disjoint, is the \emph{shuffle closure} problem. The shuffle closure of a word $U$, denoted by $U^{\odot}$, is the set of words resulting by interleaving many copies of $U$; see~\cite{DBLP:journals/jcss/WarmuthH84}. The words in $U^{\odot}$ are sometimes called \emph{shuffle powers} of $U$.

\medskip
The following problems are NP-hard for constant-sized alphabets:
\begin{enumerate}[(1)]
    \item\label{it1} Given two words $U$ and $S$, test if $S\in U^{\odot}$; see~\cite{DBLP:journals/jcss/WarmuthH84}.
    \item\label{it2} Given a word $S$, check if there exists a word $U$
    such that $|U|=|S|/2$ and $S\in U^{\odot}$ (this was originally
    called the \emph{shuffle square} problem); see \cite{DBLP:journals/jcss/BussS14}. An NP-hardness proof for a binary alphabet was recently given in \cite{DBLP:journals/tcs/BulteauV20}.
    \item\label{it3} Given a word $S$, find a shortest word $U$ such that $S\in U^{\odot}$; its hardness is trivially reduced from \eqref{it2}.
\end{enumerate}

The following observation links s-covers and shuffle closures.

\begin{observation}
Let $S$ be a word of length $n$. Then
\[S\in C^{\,\otimes}\Rightarrow \exists \, r_0,r_1\ldots, r_{n-1} \in \mathbb{Z}_+ :\;
S[0]^{r_0}S[1]^{r_1}\ldots S[n-1]^{r_{n-1}}\in C^{\,\odot}.\]
\end{observation}
\begin{proof}
Word $C$ is an s-cover of word $S$, so there exist position-subsequences $I_1,\ldots,I_k$ of length $|C|$ in $S$ that correspond to occurrences of $C$ and such that $I_1 \cup I_2 \cup \cdots \cup I_k = \{0,\ldots,n-1\}$. For $i \in \{0,\ldots,n-1\}$, we define $r_i$ as the number of occurrences of $i$ in the position-subsequences. Then, $S':=S[0]^{r_0}S[1]^{r_1}\ldots S[n-1]^{r_{n-1}}$ is a shuffle power of $C$.

To see that this holds, we transform the position-subsequences $I_1,\ldots,I_k$ into disjoint position-subsequences in $S'$ by performing the following operation simultaneously for all $i \in \{0,\ldots,n-1\}$: Replace in any way the occurrences of $i$ in $I_1,\ldots,I_k$ by distinct numbers in $\{s,\ldots,s+r_i-1\}$, where $s=r_0+\cdots+r_{i-1}$. By the definition of $r_i$, the resulting position-subsequences $I'_1,\ldots,I'_k$ in $S'$ are disjoint, satisfy $I'_1 \cup \cdots \cup I'_k = \{0,\ldots,|S'|-1\}$ and each of them corresponds to an occurrence of $C$ in $S'$.
\end{proof}

In this paper, we show that problems similar to \eqref{it1} and \eqref{it3} for s-covers, when we replace $\odot$ by $\,\otimes$, are
tractable. Notably, the first one is solved in linear time for any polynomially-bounded integer alphabet; and the last one in linear time for any constant-sized alphabet.
  
\paragraph{Our results and paper organization}
\begin{itemize}
    \item In \cref{sec:algo}, we present a linear-time algorithm for checking if a word $C$ is an s-cover of a word $S$, assuming that $S$ is over a polynomially-bounded integer alphabet $\{0,\ldots,|S|^{\cO(1)}\}$. We also discuss why an equally efficient algorithm for this problem without this assumption is unlikely.\vspace*{2mm}
    \item Let $\gamma(k)$ denote the length of a longest s-primitive word over an alphabet of size $k$. In \cref{sec:gamma,subsec:gen0,subsec:gen}, we present general bounds on this function as well as its particular values for  small values of $k$. \cref{sec:matching_lemma} contains a proof of an auxiliary lemma.
    \vspace*{2mm}
    \item In \cref{sec:slow_algo}, we show that computing a non-trivial s-cover is fixed parameter tractable for parameter $k=|\Alph(S)|$. In particular, we obtain a linear-time algorithm for computing a shortest s-cover of a word over a constant-sized alphabet.\vspace*{2mm}
    \item Finally, in \cref{sec:useless}, we explore properties of s-covers that are significantly different from properties of standard covers. In particular, we show that a word can have exponentially many different shortest s-covers.
    %, which implies that computing all shortest s-covers of a word (over a %super-constant alphabet) requires exponential %time.\footnote{\textcolor{red}{This statement assumes a non-compacted %representation of the output.}}
\end{itemize} 
We conclude in \cref{sec:concl} with a summary of our results and a list of open problems.

\section{Testing a candidate s-cover}\label{sec:algo}

Let us consider words $C=C[0\dd m-1]$ and $S=S[0\dd n-1]$. We would like to check
whether $C$ is an s-cover of $S$.

Let sequences $\FirstLex$ and $\LastLex$ be the 
lexicographically first and last 
position-subsequences of $S$ containing $C$. We assume that $\FirstLex[0]=0$ and $\LastLex[m-1]=n-1$. 
If there are no such
subsequences of positions then $C$ is not an s-cover, so we assume 
they exist and are well defined. We set a sentinel $\LastLex[m]=n$.

For all $i\in \{0,\dots, n-1\}$, we define
$$\Pref[i]= \max(\{j\,:\, \FirstLex[j]\le i\ \land\ S[\FirstLex[j]]=S[i]\}\cup\{-1\}).
$$
See also~\cref{fig:example_algo}.
Intuitively, if position $i$ is in any subsequence occurrence of $C$ in $S$, then there is a subsequence occurrence of $C$ in $S$ that consists of the prefix of $\FirstLex$ of length $\Pref[i]$, position $i$ and an appropriate suffix of $\LastLex$.
All we have to do is check, for all $i$, whether such a pair of prefix and suffix exists.
See \cref{fig:correctness} for an illustration of the argument and \cref{lem:correctness} for a formal statement of the condition that needs to be satisfied.

\begin{figure}[!ht]
    %% generated by python s_cover_tool.py -s abacbacab  -c abca
    %% string=abacbacab n=9
%% cover=abcab m=5
%% first_occ=[0, 1, 3, 5, 8]
%% last_occ=[2, 4, 6, 7, 8]
%% right=[1, 1, 2, 2, 3, 3, 4, 5, 6]
%% pref=[1, 2, 1, 3, 2, 4, 3, 4, 5]
\centering
\setlength\tabcolsep{0.4em}
%% first & last for string=abacbacab cover=abcab
\begin{tikzpicture}
\begin{scope}[xshift=0cm]
\node at (2.25,0.5) [above] {$\FirstLex$};
\node at (0.0,0) [above,blue] {\bfseries a};
\node at (0.0,0) [below,blue] {\scriptsize 0};
\node at (0.5,0) [above,blue] {\bfseries b};
\node at (0.5,0) [below,blue] {\scriptsize 1};
\node at (1.0,0) [above,] {$a$};
\node at (1.0,0) [below,] {\scriptsize 2};
\node at (1.5,0) [above,blue] {\bfseries c};
\node at (1.5,0) [below,blue] {\scriptsize 3};
\node at (2.0,0) [above,] {$b$};
\node at (2.0,0) [below,] {\scriptsize 4};
\node at (2.5,0) [above,blue] {\bfseries a};
\node at (2.5,0) [below,blue] {\scriptsize 5};
\node at (3.0,0) [above,] {$c$};
\node at (3.0,0) [below,] {\scriptsize 6};
\node at (3.5,0) [above,] {$a$};
\node at (3.5,0) [below,] {\scriptsize 7};
\node at (4.0,0) [above,blue] {\bfseries b};
\node at (4.0,0) [below,blue] {\scriptsize 8};
\end{scope}
\begin{scope}[xshift=5.5cm]
\node at (2.25,0.5) [above] {$\LastLex$};
\node at (0.0,0) [above,] {$a$};
\node at (0.0,0) [below,] {\scriptsize 0};
\node at (0.5,0) [above,] {$b$};
\node at (0.5,0) [below,] {\scriptsize 1};
\node at (1.0,0) [above,blue] {\bfseries a};
\node at (1.0,0) [below,blue] {\scriptsize 2};
\node at (1.5,0) [above,] {$c$};
\node at (1.5,0) [below,] {\scriptsize 3};
\node at (2.0,0) [above,blue] {\bfseries b};
\node at (2.0,0) [below,blue] {\scriptsize 4};
\node at (2.5,0) [above,] {$a$};
\node at (2.5,0) [below,] {\scriptsize 5};
\node at (3.0,0) [above,blue] {\bfseries c};
\node at (3.0,0) [below,blue] {\scriptsize 6};
\node at (3.5,0) [above,blue] {\bfseries a};
\node at (3.5,0) [below,blue] {\scriptsize 7};
\node at (4.0,0) [above,blue] {\bfseries b};
\node at (4.0,0) [below,blue] {\scriptsize 8};
\end{scope}
\end{tikzpicture}

\medskip

\begin{tabular}{l|cccccc}
$i$ & \scriptsize 0 & \scriptsize 1 & \scriptsize 2 & \scriptsize 3 & \scriptsize 4 & \scriptsize 5 \\ \hline
$\FirstLex[i]$ & 0 & 1 & 3 & 5 & 8 & \\
$\LastLex[i]$ & 2 & 4 & 6 & 7 & 8 & 9 \\
\end{tabular}

\medskip

%% right & pref for string=abacbacab cover=abcab
\begin{tabular}{l|ccccccccc}
$i$ & \scriptsize 0 & \scriptsize 1 & \scriptsize 2 & \scriptsize 3 & \scriptsize 4 & \scriptsize 5 & \scriptsize 6 & \scriptsize 7 & \scriptsize 8 \\ \hline
$\Pref[i]$ & 0 & 1 & 0 & 2 & 1 & 3 & 2 & 3 & 4 \\
\end{tabular}
    \caption{$\FirstLex$, $\LastLex$, and $\Pref$ arrays for $S=abacbacab$ and $C=abcab$.}\label{fig:example_algo}
\end{figure}

\begin{figure}[ht]
    \centering
\begin{tikzpicture}
    \draw[densely dashed] (1*0.6,0) .. controls (3*0.6,0.6) and (7*0.6,0.6) .. (9*0.6,0);
    \draw[densely dashed] (9*0.6,0) .. controls (10.5*0.6,0.6) and (12.5*0.6,0.6) .. (14*0.6,0);
    \draw[densely dashed] (14*0.6,0) .. controls (14.6*0.6,0.6) and (15.4*0.6,0.6) .. (16*0.6,0);
    \draw[densely dashed,xshift=1.2cm] (14*0.6,0) .. controls (14.6*0.6,0.6) and (15.4*0.6,0.6) .. (16*0.6,0);
    \draw[densely dashed,xshift=2.4cm] (14*0.6,0) .. controls (14.6*0.6,0.6) and (15.4*0.6,0.6) .. (16*0.6,0);
    \foreach \c/\v [count=\i] in {red/$p_0$,black/,red/$p_1$,black/,green/$q_0$,red/$p_2$,black/,black/,blue/$i$,red/$p_3$,green/$q_1$,black/,red/$p_4$,green/$q_2$,black/,green/$q_3$,red/$p_5$,green/$q_4$,black/,green/$q_5$}
    {
        \filldraw[color=\c] (\i*0.6,0) circle (0.07cm);
        \draw (\i*0.6,0) node[below] {\v};
    }
\end{tikzpicture}
    \caption{Assume that for some words $C$ and $S$ the sequences $\FirstLex=\{p_0,p_1,p_2,p_3,p_4,p_5\}$ (red) and $\LastLex=\{q_0,q_1,q_2,q_3,q_4,q_5,|S|\}$ (green) are as in the figure. Further assume that $\Pref[i]=1$ (i.e., we have $S[i]=S[p_1]\neq S[p_2]$). As shown, we have $q_{\Pref[i]+1}=q_2>i$. Thus, the position $i$ is covered by an occurrence of $C$ as a subsequence using positions $(p_0,i,q_2,q_3,q_4,q_5)$.} 
    \label{fig:correctness}
\end{figure}
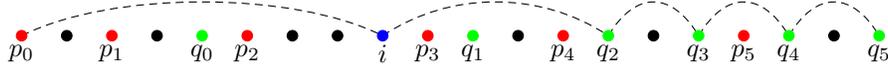

The sequence $\FirstLex$ can be computed with a simple left-to-right pass over $S$ and $C$; the computation of $\LastLex$ is symmetric. The table $\Pref[i]$ is computed on-line using an additional table $\PRED$ indexed by the letters of the alphabet. The algorithm is formalized in the following pseudocode. 

\medskip
\begin{algorithm}[H]
\SetKwBlock{Begin}{}{}
\caption{$\mathit{TEST}(C,S)$}

\smallskip
\KwIn{
    word $C=C[0\dd m-1]$ and
    word $S=S[0\dd n-1]$
}
\KwOut{$true$ if and only if $C$ is an s-cover of $S$}
\smallskip
compute $\FirstLex$ and $\LastLex$ \;

\smallskip
\Begin(\LeftComment{compute $\Pref$}\vspace*{2mm}){
    \lForAll{$c \in \Sigma$}{$\PRED[c]:=-1$}
    $k:=0$ \;
    \For{$i:=0$ \KwTo $n-1$}{
        \If{$i=\FirstLex[k]$}{
            $\PRED[S[i]]:=k$ \;
            \lIf{$k<m-1$}{$k := k+1$}
        }
        $\Pref[i]:=\PRED[S[i]]$ \;
    }
}

\medskip
\KwRet{
    $\forall_{i=0,\ldots,n-1}\,
    (\,\Pref[i] \ne -1\ \textrm{\bf and}\ 
    \LastLex[\Pref[i]+1] > i\,)$
}
\vspace*{0.2cm}
\label{algo:TEST}
\end{algorithm}

\smallskip
Correctness of the algorithm follows from \cref{lem:correctness} (inspect also~\cref{fig:correctness}). 

\begin{lemma}\label{lem:correctness}
Let us assume that $\FirstLex$ and $\LastLex$ are well defined.
Then 
$C$ is an s-cover of $S$ if and only if for each position $0\le i\le n-1$ we have
$\Pref[i]\ne -1$ and $\LastLex[\Pref[i]+1]>i$.
\end{lemma}
\begin{proof}
First, observe that if $\Pref[i]=-1$ for any $i$, then $C$ is not an s-cover of~$S$. This follows from the greedy computation of $\FirstLex$, which implies that the prefix of $C$ that precedes the first occurrence of $S[i]$ in $C$ does not have a subsequence occurrence in $S[0 \dd i-1]$; else, $i$ would be in $\FirstLex$, a contradiction.

We henceforth assume that $\Pref[i]\ne -1$ for every $i$ and show that, in this case, $C$ is an $s$-cover of $S$ if and only if $\LastLex[\Pref[i]+1]>i$ for all $i \in \{0, \ldots , n-1\}$. Let $\FirstLex=\{p_0,p_1,\ldots,p_{m-1}\}$ and $\LastLex=\{q_0,q_1,\ldots,q_{m-1},q_m=n\}$.

\paragraph{$\mathbf{(\Leftarrow)}$}
Assume that $\LastLex[\Pref[i]+1]>i$. 

In this case, position $i$ can be covered by a subsequence occupying positions $p_0,\ldots,p_{j-1},i,q_{j+1},\ldots,q_{m-1}$, for $j=\Pref[i]$. As $S[p_j]=S[i]$, this subsequence is equal to $C$, and as $p_j\le i$ (by the definition of $\Pref[i]$) and $q_{j+1}>i$ (by the assumption), those positions form an increasing sequence (that is, we obtain a valid subsequence).

\paragraph{$\mathbf{(\Rightarrow)}$}
Assume that for some $j$ there exists an increasing sequence
\[r_0,r_1,\ldots,r_{j-1},i,r_{j+1},\ldots,r_{m-1},\] 
such that $S[r_0]S[r_1]\ldots S[r_{j-1}]S[i]S[r_{j+1}]\ldots S[r_{m-1}]=C$.

By induction for $k=0,\ldots, j$, $r_k\ge p_k$ (including $r_j=i$) and for $k=m-1,\ldots, j+1$, $r_k\le q_k$.
But this means that $\Pref[i]\ge j$. If $j<m-1$, we have $i<r_{j+1}\le \LastLex[j+1]$, and otherwise, $i<\LastLex[j+1]=n$.
Hence, $\LastLex[\Pref[i]+1]\ge \LastLex[j+1]>i$.
This completes the proof.
 \end{proof}
Note that, under the assumption of a polynomially-bounded integer alphabet, the table $\PRED$ can be initialized and updated deterministically in linear total time by first sorting the letters of $S$ using Radix Sort.
We thus arrive at the following result.

\begin{theorem}\label{lem:checking_alg}
Given words $C$ and $S$ over an integer alphabet $\{0,\ldots,|S|^{\cO(1)}\}$, we can check if $C$ is an s-cover of $S$ in $\cO(|S|)$ time.
\end{theorem}

In the standard setting (cf.~\cite{DBLP:journals/ipl/Breslauer92}), one can check if a word $C$ is a cover of a word~$S$---what is more, find the shortest cover of $S$---in linear time for any (non-necessarily integer) alphabet. We show below that the existence of such an algorithm for testing a candidate s-cover is rather unlikely. 
Let us introduce a slightly more general version of the s-cover testing problem in which, if $C$ is an s-cover of $S$, we are to say, for each position $i$ in $S$, which position $j$ of $C$ is actually used to cover $S[i]$; if there is more than one such position $j$, any one of them can be output. Let us call this problem the \emph{witness s-cover testing} problem. In particular, our algorithm solves the witness s-cover testing problem with the answers stored in the $\Pref$ array. %Actually it is hard to imagine an algorithm that solves the s-cover testing problem and not the witness version of it.
We next give a comparison-based lower bound for the latter.

\begin{theorem}
The witness s-cover testing problem for a word $S$ of length $n$ (and a word $C$) requires $\Omega(n \log n)$ time in the comparison model.
\end{theorem}
\begin{proof}
Let us consider a word $C$ of length $m$ that is composed of $m$ distinct letters and a family of words of the form $S=CTC$, where $T$ is a word of length $m$ such that $\Alph(T)\subseteq\Alph(C)$. Then $C$ is an s-cover of each such word $S$. Each choice of the word $T$ implies a different output to the witness s-cover testing problem on $C$ and $S$. There are $m^m$ different outputs, so a decision tree for this problem must have depth $\Omega(\log m^m) = \Omega(m \log m) = \Omega(n \log n)$. 
\end{proof}

We can define the \emph{coverage} of a word $C$ in a word $S$ as the number of positions in $S$ that are covered by occurrences of $C$ in $S$ as subsequences. If $C$ is not a subsequence of $S$, the coverage is equal to 0. If $C$ is an s-cover of $S$, the coverage is equal to $|S|$.

Even if $C$ turns out not to be an s-cover of $S$, Algorithm~\ref{algo:TEST} actually computes the positions of $S$ that can be covered using occurrences of $C$ (they are exactly the positions $i$ for which $\Pref[i]\ne -1$ and $\LastLex[\Pref[i]+1]>i$). This leads to the following corollary.

\begin{corollary}
Given words $C$ and $S$ over an integer alphabet $\{0,\ldots,|S|^{\cO(1)}\}$, we can compute the coverage of $C$ in $S$ in $\cO(|S|)$ time.
\end{corollary}

Hence, our algorithm can be useful to find partial variants of s-covers, defined analogously as for the standard covers~\cite{DBLP:journals/tcs/FlouriIKPPST13,DBLP:journals/algorithmica/KociumakaPRRW15,DBLP:conf/esa/Radoszewski23}.

\section{Maximal lengths of s-primitive words over small  alphabets}\label{sec:gamma}
Let us recall that $\gamma(k)$ denotes the length of a longest s-primitive word over an alphabet of size $k$. It is obvious that $\gamma(1)=1$ and $\gamma(2)=3$ since there are no square-free words of length larger than $3$ over a binary alphabet. In particular, the longest s-primitive binary words are $aba$ and $bab$. 

The case of ternary words is already more complicated; in this section, we study this case and the case of a quaternary alphabet. General bounds on the function $\gamma(k)$ are shown in \cref{subsec:gen} (weaker but simpler bounds are presented in \cref{subsec:gen0}). %A discussion on computing $\gamma(k)$ for small $k>3$ is presented in \cref{subsec:smallk}. In particular, we were not able to compute the exact value of $\gamma(5)$.

%\subsection{Ternary and quaternary alphabet}\label{subsec:3}
%
%The computation of $\gamma(k)$ is nontrivial even 
%in the cases of alphabets of size 3 and 4.

\begin{observation}\label{obs:scover_factor}
If a factor of the word $S$ is not s-primitive, then $S$ is also not
s-primitive.
\end{observation}
\begin{proof}
If $S=PUQ$ and $C$ is a non-trivial s-cover of $U$, then $PCQ$ is a non-trivial s-cover of $S$.
\end{proof}

\begin{fact}\label{fact}
$\gamma(3)=8$.
\end{fact}
\begin{proof} 
The word $S=abcabacb$ is of length 8 and it is s-primitive.
Indeed, every s-cover of $S$ needs to have a prefix $abc$, a suffix $acb$, and no further letters $c$.
The only word of the form $abcXacb$, for word $|X| \le 2$ over the alphabet $\{a,b\}$, that is an s-cover of $S$ is $S$.
Hence, $\gamma(3) \geq 8$.

In order to show that this bound is tight, we still have to show that each ternary word of length 9 is not s-primitive.
(There are 19683 ternary words of length 9). 
The number of words to consider is substantially reduced by observing that 
relevant words are square-free and do not contain the structure
specified in the following claim.

\begin{claim}\label{obs:square_or_abXbc}
If a word $S$ over a ternary alphabet contains a factor of
the form $abXbc$ for some (maybe empty) word $X$ and different letters $a,b,c$, then it is not s-primitive.
\end{claim}
\begin{proof}
The factor $abXbc$ has $abc$ as its s-cover, and thus it is not s-primitive. Consequently, by \cref{obs:scover_factor}, the whole word $S$ is not s-primitive.
\end{proof}

\begin{figure}[!ht]
    \centering
    \tikzset{font=\small}
    \begin{tikzpicture}[yscale=0.56,xscale=0.55,auto,node distance=0.1cm]

\newcommand{\red}[1]{\textcolor{red}{#1}}

\tikzstyle{dot}=[inner sep=0.03cm, circle, draw, fill=gray]
\node[dot] (r) at (0,0) {};
\foreach \name/\dy/\parent/\l in {
  a/0/r/a,
  ab/0/a/b,
  aba/3.7/ab/a,
  abac/0/aba/c,
  abaca/1/abac/a,
  abacab/0/abaca/b,
  abacaba/0.5/abacab/a,
  %abacabac/0/abacaba/c,
  abacabc/-0.5/abacab/c,
  abacb/-1/abac/b,
  abacba/0.5/abacb/a,
  abacbab/0/abacba/b,
  abacbabc/0/abacbab/c,
  abacbc/-0.5/abacb/c,
  abc/-3.7/ab/c,
  abca/1.9/abc/a,
  abcab/1.7/abca/b,
  abcaba/0/abcab/a,
  abcabac/0/abcaba/c,
  abcabaca/0.8/abcabac/a,
  abcabacb/-0.8/abcabac/b,
  abcabacba/0.5/abcabacb/a,
  abcabacbc/-0.5/abcabacb/c,
  abcac/-1.7/abca/c,
  abcacb/0/abcac/b,
  abcacba/0.5/abcacb/a,
  abcacbab/0.5/abcacba/b,
  abcacbac/-0.5/abcacba/c,
  abcacbaca/0/abcacbac/a,
  abcacbc/-0.5/abcacb/c,
  abcb/-1.9/abc/b,
  abcba/0/abcb/a,
  abcbab/0.7/abcba/b,
  abcbabc/0/abcbab/c,
  abcbac/-0.7/abcba/c,
  abcbaca/0.5/abcbac/a,
  abcbacb/-0.5/abcbac/b,
  abcbacbc/0/abcbacb/c
}{
  \node[dot] (\name) at ($(\parent)+(2,\dy)$) {};
  \draw[black,thick] (\parent)--(\name) node[midway,auto] {\textcolor{blue}{\l}};
}

\foreach \i/\place/\word in {
%abacabac/a\red{ba}cab\red{ac},
1/abacabc/\red{ab}aca\red{bc},
2/abacbabc/\red{ab}acba\red{bc},
3/abacbc/\red{ab}ac\red{bc},
4/abcabaca/a\red{bc}aba\red{ca},
7/abcacbab/ab\red{ca}cb\red{ab},
9/abcacbc/\red{ab}cac\red{bc},
10/abcbabc/\red{ab}cba\red{bc},
11/abcbaca/a\red{bc}ba\red{ca},
12/abcbacbc/\red{ab}cbac\red{bc}
}{
\node (w\i) [] at ($(\place)+(1.5,0)$) {\small \word};
}
\foreach \i/\place/\word in {
%abacabac/a\red{ba}cab\red{ac},
5/abcabacba/abcabacba,
6/abcabacbc/\red{ab}cabac\red{bc},
8/abcacbaca/a\red{bc}acba\red{ca}
}{
\node (w\i) [] at ($(\place)+(2,0)$) {\small \word};
}

\draw ($(w5.center)-(1.8,0.4)$) rectangle ($(w5.center)+(1.8,0.4)$);
\draw ($(w6.center)-(1.8,0.4)$) rectangle ($(w6.center)+(1.8,0.4)$);
\draw ($(w8.center)-(1.8,0.4)$) rectangle ($(w8.center)+(1.8,0.4)$);

\filldraw[green!50!black] (abcabacb) circle (5pt) {};
\filldraw[green!50!black] (abcacbac) circle (5pt) {};

\end{tikzpicture}
\caption{
A trie of all ternary square-free words starting with $ab$, truncated at words that are not s-primitive (in leaves).
Only one word in a leaf ($abcabacba$) does not contain the structure specified in the claim inside the proof of \cref{fact}, but it still has a non-trivial s-cover $abcba$.
The trie has depth 9 (the leaves with words of length 9 are shown in frames and the internal nodes of depth 8 corresponding to s-primitive words are drawn as big dots), so $\gamma(3)=8$.
}\label{fig:ternary}
\end{figure}
\cref{fig:ternary} shows a trie of all ternary square-free words starting with $ab$, truncated at words that are not s-primitive (in leaves). The words in all leaves but one contain the structure from the claim, and for the remaining word, a non-trivial s-cover can be easily given. 
The trie shows that words of length $9$ over a ternary alphabet are not s-primitive.
 \end{proof}

%For a quaternary alphabet we used computer experiments.

\begin{fact}\label{Kuba's mail}
$\gamma(4)=19$.
\end{fact}
\begin{proof}
Using a computer-assisted verification we have checked that the word $S=abacadbabdcabcbadac$, of length 19, is s-primitive. Thus $\gamma(4)\ge 19$.
We also checked experimentally that $\gamma(4)\le 19$. 
Consequently, $\gamma(4)=19$.
\end{proof}

An optimized \texttt{C++} code used for the experiments can be found at \url{https://github.com/craniac-swistak/s-covers}. The program reads $k$ and computes $\gamma(k)$.

\begin{remark}
There are $2 \cdot 3! = 12$ s-primitive words of length $\gamma(3)=8$ over ternary alphabet (cf.\ \cref{fig:ternary}, for each pair of distinct letters there are two s-primitive words starting with these letters). This accounts for less than $0.2\%$ among all $3^8$ ternary words of length 8.
For a 4-letter alphabet, our program shows that the relative number of s-primitive words of length $\gamma(4)=19$ is very small. There are exactly $2496$ such words, out of $4^{19}$, which gives a fraction less than $10^{-8}$.
This suggests that s-primitive 5-ary words of length $\gamma(5)$ are extremely sparse and
finding an s-primitive word over a 5-letter alphabet of length $\gamma(5)$ could be a challenging task. (In particular, in the next subsection we show that $\gamma(5)\ge 39$.)
\end{remark}

\cref{tab:gamma_ex} summarizes this section.
\begin{table}[htbp]
    \centering
    \begin{tabular}{c|c|l}
         $k$ & $\gamma(k)$ & $\quad$examples of longest s-primitive words \\\hline\hline
         $\quad$1$\quad$ & 1 & $\quad a$ \\\hline
         2 & 3 & $\quad aba$ \\\hline
         3 & 8 & $\quad abcabacb$ \\\hline
         \multirow{1}{*}{4} & \multirow{1}{*}{$\quad$19$\quad$} & $\quad abacadbabdcabcbadac$ %\\
         %&& $\quad abcdabacadbdcbabdac$
    \end{tabular}
    \vspace*{0.2cm}
    \caption{The values of $\gamma$ for small alphabet size $k$.}
    \label{tab:gamma_ex}
\end{table}

\section{General alphabet, preliminary  upper  bound}\label{subsec:gen0}
We start with a simple proof of a preliminary upper bound.
We use two characteristic subsequences defined below.
For a word $S$ over alphabet $\Alph(S)$ of size $k$, let $\first(S)$ (resp.\ $\last(S)$) denote the length-$k$ word containing all the letters of $\Alph(S)$ in the order of their first (resp.\ last) occurrence in $S$. 

\begin{example}We have
\[\first(\underline{ab}a\underline{d}b\underline{c}d)=abdc,\ \last(ab\underline{a}d\underline{bcd})=abcd.\]
\end{example}

\begin{lemma}\label{lem:FL}
For a word $S$, let $F=\first(S),\ L=\last(S)$.
If the word $FLFL$ is a subsequence of $S$, then $FL$ is an s-cover of $S$. 
\end{lemma}
\begin{proof}
We need to show that each position $i$ of $S$ is covered by an occurrence of $FL$ as a subsequence. 

There exists a position $j$ in $S$ such that $FL$ is a subsequence of each of the words $S[\dd j]$ and $S(j\dd ]$. We can assume that $i \le j$; the other case is symmetric.

Let $p$ be the index such that $F[p]=S[i]$. It suffices to argue that:

\begin{enumerate}[(1)]
    \item\label{aa} $F[\dd p)$ is a subsequence of $S[\dd i)$
    \item\label{bb} $(FL)(p\dd ]$ is a subsequence of $S(i\dd ]$.
\end{enumerate}

\noindent Point  \eqref{aa} follows by the definition of $F=\first(S)$. As for point \eqref{bb}, $S(i\dd ]$ has a subsequence $FL$ by the fact that $i \le j$ and the definition of $j$, and $(FL)(p\dd ]$ is a suffix of $FL$.
\end{proof}

The next observation follows from \cref{obs:scover_factor} and the definition of $\gamma(k)$.

\begin{observation}\label{obs:gamma_k_plus_1}
If word $S$ is s-primitive and $|S| > \gamma(k-1)$, then every factor of $S$ of length $\gamma(k-1)+1$ contains at least $k$ different letters.
\end{observation}

\begin{fact}\label{fct:ub}
For $k>1$, $\gamma(k)\le (2k+2)\, (\gamma(k-1)+1)-1$.
\end{fact}
\begin{proof}
It is enough to show that if a word $S$ satisfies $|\Alph(S)|=k$ and $|S|=(2k+2)\, (\gamma(k-1)+1)$, then $S$ is not s-primitive. Assume to the contrary, that such a word $S$ is s-primitive. Let $F=\first(S),\ L=\last(S)$. We will show that $FLFL$ is a subsequence of $S$. By \cref{lem:FL}, this will yield a contradiction as $|FL|=2k<|S|$.

Let us partition $S$ into blocks $S_1,S_2,\ldots,S_{2k+2}$ of length $\gamma(k-1)+1$. By \cref{obs:gamma_k_plus_1}, each of the blocks contains all the $k$ different letters of $S$. 

Hence, $F$ is a subsequence of $S_1$, $L$ (as a length-$k$ word) is a subsequence of $S_2 S_3 \cdots S_{k+1}$, $F$ is a subsequence of $S_{k+2} S_{k+3} \cdots S_{2k+1}$, and $L$ is a subsequence of $S_{2k+2}$. Consequently, $FLFL$ is a subsequence of $S$.
\end{proof}

\begin{comment}
\begin{observation}\label{obs:one_letter}
If a letter $a$ occurs in a word $S=S' a S''$ only once, then $C$ is an
s-cover of $S$ if and only if  $C\,=\,C' a C''$, where $C',C''$ are s-covers of $S',S''$, respectively.
\end{observation}

\begin{fact}\label{fct:lb}
$\gamma(k)\ge 2^k-1$.
\end{fact}
\begin{proof}
Let $a_1,a_2,\ldots$ be an 
infinite sequence of distinct letters.
We consider the sequence of Zimin words, defined in Equation~\eqref{zimin}.
By \cref{obs:one_letter}, each $Z_k$ is s-primitive.
\end{proof}
\end{comment}

In the next section, we show how to improve the preliminary bound on $\gamma(k)$ from \cref{fct:ub}.% and \cref{fct:lb}.

\section{\bf General alphabet, lower bound and improved upper bound}\label{subsec:gen}
\subsection{Lower bound}
\noindent We need the  following straightforward fact.

\begin{observation}\label{obs:one_letter}
If a letter $a$ occurs in a word $S=S' a S''$ only once, then $C$ is an
s-cover of $S$ if and only if  $C\,=\,C' a C''$, where $C',C''$ are s-covers of $S',S''$, respectively.
\end{observation}

First we obtain a recurrence showing the relation between $\gamma(k)$
and $\gamma(k-1)$ in general case. Then we use the word from \cref{Kuba's mail} as a base of the recurrence.

\begin{lemma}\label{thm:Tomek 29.11}
 For $k \geq 2$ we have $\gamma(k) \ge 2\cdot\gamma(k-1)+1$.
% 
%\
%\le \Delta(k).%2^{k-1}\,k!.
 %\]
 \end{lemma}
\begin{proof}
Let $S$ be an s-primitive word of length $\gamma(k-1)$ with $|\Alph(S)|=k-1$ and $a \not\in\Alph(S)$. Then, by \cref{obs:one_letter}, $S'=SaS$ is s-primitive. We have $|\Alph(S')|=k$ and $|S'|=2\cdot\gamma(k-1)+1$.
\end{proof}

\begin{theorem}[Lower bound]\label{cor:Tomek 29.11}
    For $k \geq 4$ we have $\gamma(k) \ge 5\cdot 2^{k-2}-1$.
\end{theorem}
\begin{proof}
Due to \cref{Kuba's mail}, $\gamma(4)=19=5\cdot 2^{4-2}-1$. By \cref{thm:Tomek 29.11} and induction, we have for $k>4$,
\[ \ \gamma(k)\ge 2\cdot\gamma(k-1)+1 \ge 2\cdot (5\cdot 2^{k-3}-1)+1 =  5\cdot 2^{k-2}-1.\qedhere\]  
\end{proof}

 \subsection{Two combinatorial facts}
 We define the following condition for two words $X,Y$ of length at least 2.
\noindent 
\begin{description}
    \item{$\Phi(X,Y)$:\ }
 there exist factors $a_1a_2$ of $X$ and $b_1b_2$ of $Y$, where $a_1,a_2,b_1,b_2$ are letters,  satisfying
 \begin{equation}\label{eq:phi}
 a_1\notin \{b_1,\, b_2\}\ \land\ b_2\notin \{a_1,\, a_2\}.
 \end{equation}
\end{description}
%We also say that $X$ ``matches'' $Y$.

\begin{example} We have $\Phi(X,Y)$ for $X=a\underline{bc}a,\ Y=aa\underline{cd}dd$.
\end{example}

The proof of the following technical lemma is deferred to \cref{sec:matching_lemma}. 
%As a side note, the lemma can be verified experimentally by a %computer. It is enough to check all words 
%of length 10 over alphabet $\{1,\ldots,10\}$.

\begin{lemma}[XY-Lemma]
\label{lem:comb}\mbox{ \ }\\
If $X,Y$ are square-free words of length 4 and 6, respectively, then $\Phi(X,Y)$.
\end{lemma}

For a word $S=a_1a_2\dd a_n$ of length $n\ge 4$ we introduce  yet another, similar condition: \[\Psi(S) \ \Leftrightarrow\ \Phi(a_1a_2,\ a_{n-1}a_n).\]
We assume that $\Psi(S)$ holds if $n\le 3$. 

\begin{example}% The word $S=abaaaabc$ satisfies the condition $\Psi(S)$.
We have
%$\Phi(ab,cb)$ but not $\Phi(ab,bc)$, and 
$\Psi(abbc)$ but not $\Psi(abcb)$.
\end{example}

%\begin{lemma}\label{lem: cutting in the word}
%If $W$ is square-free and $|W| \ge 12$, then $W$ contains a factor $W'$  such that $|W'|\ge |W|-5$ and $\Psi(W')$.
%\end{lemma}
%\begin{proof}
%I
%\end{proof}

The following fact is a consequence of the previous lemma; see \cref{fig:reduction}.

\begin{figure}[!bt]
    %% generated by python s_cover_tool.py -s abacbacab  -c abca
    \centering
    \begin{tikzpicture}[scale=0.5,transform shape]
\tikzset{font=\Large,color=blue};
\def\boxH{0.8}
\draw (0,2) rectangle +(11.5, \boxH) ;
\draw (1,0) rectangle +(9, \boxH);

\node[left] at (0,2+0.5*\boxH) {$W$};
\node[left] at (1,0.5*\boxH) {$W'$};

\foreach \x/\label in {1/U, 9/V}
 {
  \draw[fill=white!85!yellow] (\x,0) rectangle +(1,\boxH);
  \draw[fill=white!85!yellow] (\x,2) rectangle +(1,\boxH) node[midway] {$\label$};
}

\draw[latex-latex] (0,3)--+(3,0) node[midway,above] {$X$};
\draw[dashed] (3,2)--+(0,\boxH);

\draw[latex-latex] (11.5,3)--+(-4,0) node[midway,above] {$Y$};
\draw[dashed] (7.5,2)--+(0,\boxH);

\end{tikzpicture}
    \caption{Illustration of  Corollary~\ref{lem: cutting in the word2}. 
    $|X|=4$, $|Y|=6$ and $|U|=|V|=2$, consequently $|W'|\ge |W|-6$. If $\Phi(U,V)$  then  $\Psi(W')$.}\label{fig:reduction}
\end{figure}
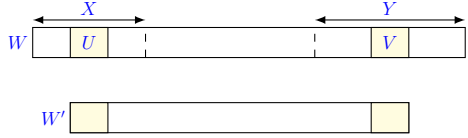

\begin{corollary}\label{lem: cutting in the word2}
If $W$ is square-free, then $W$ contains a factor $W'$  such that\linebreak $|W'|\ge |W|-6$ and $\Psi(W')$.
%and $X'[0] \not\in \{Y'[0],Y'[1]\}$, $Y'[0] \not\in \{X'[0],X'[1]\}$.
\end{corollary}
%\begin{comment}
\begin{proof} 
We assume that $|W|\ge 10$ since for $|W|<10$, $\Psi(W')$ is true for $W'=W[0\dd \min(3,|W|))$.

If $|W|\ge 10$, then consider the prefix $X$ of $W$ of length 4 and the suffix $Y$ of $W$ of length 6; see \cref{fig:reduction}. By \cref{lem:comb}, $\Phi(X,Y)$ holds. Let $U=a_1a_2$ and $V=b_1b_2$ be the factors of $X$ and $Y$, respectively, that satisfy the corresponding condition \eqref{eq:phi}, with $U=X[i \dd i+1]$ and $V=Y[j \dd j+1]$. Then $\Psi(W')$ holds for the factor $W'=W[i \dd (|W|-6)+j+1]$; we have $|W'|=|W|-4+j-i \ge |W|-6$ as $i \le 2$, $j \ge 0$.
 \end{proof}
 %\end{comment}

\begin{example} The number 6 in the above corollary cannot be decreased.
The word $W=abcabadabacba$ is square-free (and s-primitive, and a palindrome). For every factor $W'$ of $W$ of length $|W'|\ge |W|-5$, $\Psi(W')$ does not hold. Let us note that $W''=badabac$ is a factor of $W$ of length $|W|-6$ such that $\Psi(W'')$ holds.
\end{example}

\subsection{Upper bound}

In this section, we give a refined upper bound $\gamma(k)\le (2k-2)\gamma(k-1)+6$ for $k\ge 4$.

To prove it we need to show that if a word $S$ satisfies $|\Alph(S)|=k$ and $|S|>(2k-2)\gamma(k-1)+6$, then $S$ is not s-primitive. By \cref{lem: cutting in the word2} and \cref{obs:scover_factor} we know that it is enough to prove this fact for $S$ such that $\Psi(S)$ and $|S|> (2k-2)\gamma(k-1)$ (a factor of the original $S$).

%Our proof is by contradiction. The general idea is to %show  the implication 
%
%\medskip\centerline{
% \framebox{ $ \,S\mbox{ is s-primitive and }|S|>(2k-%2)\gamma(k-1)\, \ \Rightarrow\ \D\mbox{ is an s-cover %of }S.$}
% }
%that leads to a contradiction.

%\medskip\noindent 
%We assume, until the end of this section, that  
%$S$ is s-primitive.

Let $F=\first(S)$. If the last letter of $F$ and the first letter of $\last(S)$ are different, then let
$L=\last(S)$, otherwise let $L$ equal $\last(S)$ with the first letter deleted.

\begin{observation}\label{obs:no2}
    $FL$ does not contain any equal adjacent letters.
\end{observation}

Words $F,L$ have the following useful property.

\begin{observation}\label{obs:March14}
    If $S[i]$ equals the $j$th letter of $F$, then  there is a 
sequence of $j$ positions  $i_1<i_2<\cdots<i_j=i$ such that $S[i_1]S[i_2]\cdots S[i_j]$ is a prefix of $F$. A symmetric ``reverse'' property holds for $L$.
\end{observation}

\begin{lemma}\label{April12_2024}
If $\Psi(S)$ and  $S$ is  s-primitive 
then $|S|\le (2k-2)\gamma(k-1)$.
\end{lemma}
\begin{proof}
The proof is by contradiction.
We show that
if
$S$ is s-primitive and $|S|> (2k-2)\gamma(k-1)$, then $FL$ is an s-cover of $S$.

Let us cover $S$ with $2k-2$ blocks, each of length $\gamma(k-1)+1$, with overlaps of one position between consecutive blocks.
The final block can be longer if $|S|> (2k-2)\cdot \gamma(k-1)+1$.

Let $B_1,B_2,\ldots,B_{2k-2}$ denote the respective blocks and $\mathbf{T}=B_1B_2 \cdots B_{2k-2}$ be their concatenation. \cref{obs:no2} implies the following.

\begin{observation}\label{123}
$FL$ is an s-cover of $S$ if and only if it is an s-cover of $\mathbf{T}$.
\end{observation}
\begin{proof}
Word $\mathbf{T}$ is the word $S$ with some letters duplicated.
If $FL$ is an s-cover of $S$, then the duplicates of letters that occur in $\mathbf{T}$ can be covered using the same subsequences as the originals.
Conversely, if $FL$ is an s-cover of $\mathbf{T}$, each subsequence of $\mathbf{T}$ equal to $FL$ transforms to a subsequence equal to $FL$ in $S$ by removing the duplicated letters.
Indeed, by \cref{obs:no2}, a subsequence equal to $FL$ of $\mathbf{T}$ uses at most one of each pair of equal adjacent letters.
\end{proof}

Therefore, it suffices to show that $FL$ is an s-cover of $\mathbf{T}$.
By \cref{obs:gamma_k_plus_1}, each block in $\mathbf{T}$ contains all the $k$ letters.
The first and the last block in $\mathbf{T}$ contain $F$ and $L$ as a subsequence, respectively.
Any $k-3$ consecutive blocks contain $F_-$ or $L^-$ as a subsequence, where $X_-$ (respectively $X^-$) is the word obtained from $X$ by deleting the first (last, respectively) \textbf{three letters}.

Hence there is a decomposition $\mathbf{T}\,=\, A\,B\,UV\,C\,D$, where $U,V$ are  \emph{middle} blocks,  $A$ is the first block and contains $F$, $B$ contains $L^-$, $U$ and $V$ each contain every distinct letter, $C$
contains $F_-$, and $D$ contains $L$.

Assume we want to cover a position $i$ with an occurrence of $FL$ as a subsequence. We start by showing that each position $i$ inside the $UV$ block is covered.
It is enough to  prove the following claim for the simplified version $\mathbf{T}'$ of $\mathbf{T}$ in which  the outer blocks are replaced only by the letters used when a position inside $UV$ is covered with an occurrence of $FL$.

We notice that $\Psi(\mathbf{T})\Leftrightarrow \Psi(\mathbf{T}')$, since the first two and last two letters in both words match.

\begin{figure}[ht]
%\vspace*{-0.2cm}
\centering
    \begin{tikzpicture}[scale=0.87,transform shape]

\tikzstyle{dot}=[inner sep=0.045cm, circle, draw, fill=red]

\newcommand{\casesTemplate}[0]{
  \foreach \i in {1,...,6} {
    \node[above] (la\i) at (\i*0.5, 0) {$a_\i$};
  }
  \foreach \i in {1,...,3} {
    \node[above] (lb\i) at (3+\i*0.5, 0) {$b_\i$};
  }
  \draw[blue] (5,0) rectangle +(2,0.5) node[midway] (u) {};
  \draw[red] (7,0) rectangle +(2,0.5) node[midway] (v) {};
  \foreach \i in {4,...,6} {
    \node[above] (ra\i) at (7.5+\i*0.5, 0) {$a_\i$};
  }
  \foreach \i in {1,...,6} {
    \node[above] (rb\i) at (10.5+\i*0.5, 0) {$b_\i$};
  }
}

\begin{scope}
\casesTemplate
\node[green!40!black] (ua) at ($(u)-(0.2,0)$) {$a_1$};
\foreach \p in {la1,la2,la3,la4,la5,la6,lb1,lb2,ua,rb4,rb5,rb6} {
  \node[dot] at ($(\p.south)-(0,0.2)$) {};
}
\draw[orange,thick] (ua) circle (0.22cm);
\node at ($(u)+(0,0.5)$) {$U$};
\node at ($(v)+(0,0.5)$) {$V$};
\node at ($(u)+(1,-0.5)$) {$a_1=b_3$};
\end{scope}

\begin{scope}[yshift=-2cm]
\casesTemplate
\node[black] (ub4) at ($(u)-(0.2,0)$) {$b_4$};
\node[green!40!black] (va2) at ($(v)+(0.2,0)$) {$a_2$};
\draw[orange,thick] (va2) circle (0.22cm);
\foreach \p in {la1,la2,la3,la4,la5,la6,lb1,lb2,lb3,ub4,va2,rb6} {
  \node[dot] at ($(\p.south)-(0,0.2)$) {};
}
\node at ($(u)+(1,-0.5)$) {$a_2=b_5$};
\end{scope}

\begin{scope}[yshift=-4cm]
\casesTemplate
\node[green!40!black] (ua2) at ($(u)-(0.2,0)$) {$a_2$};
\draw[orange,thick] (ua2) circle (0.22cm);
\node[black] (va3) at ($(v)+(0.2,0)$) {$a_3$};
\foreach \p in {la1,ua2,va3,ra4,ra5,ra6,rb1,rb2,rb3,rb4,rb5,rb6} {
  \node[dot] at ($(\p.south)-(0,0.2)$) {};
}
\end{scope}

\end{tikzpicture}
    \caption{We have here $\Psi(\mathbf{T}')$, where  $\mathbf{T}'\,=\, F\cdot L^-\cdot U\cdot V\cdot F_-\cdot L$,  
    $F=a_1 \cdots a_6$, $L=b_1 \cdots b_6$. 
    In the construction, we use the fact that $a_1\ne b_5$, $a_1\ne b_6$ and $a_2\ne b_6$. The big dots
    show the $\D$-subsequences covering a chosen position $i$ in $UV$ containing $a_1$ or $a_2$. In the last two cases, we find an \emph{additional} position in $U,V$ containing $b_4$, $a_3$, respectively.  
    The  case $\mathbf{T}'[i]=b_6 \vee \mathbf{T}'[i]=b_5$   is symmetric.}\label{fig:3cases}
\end{figure}
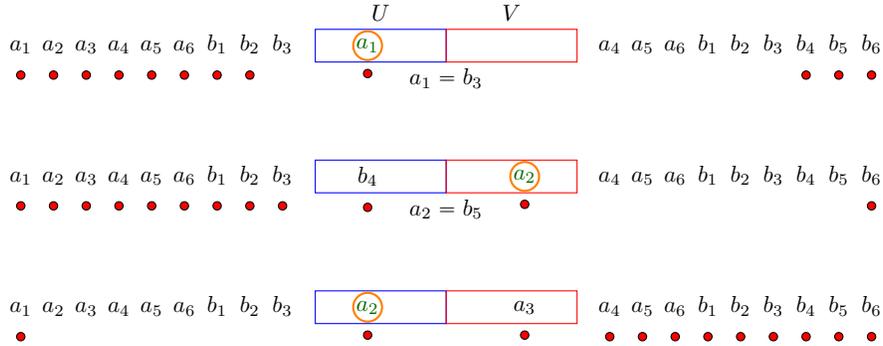

\begin{claim}[Middle-Blocks]\label{lem: critical cases}
Let $k\ge 4$ and $\mathbf{T}'\,=\, F\cdot L^-\cdot U\cdot V\cdot F_-\cdot L$, where $F$ is a permutation of $\Alph(\mathbf{T}')$ and either $L$ or $F[|F|-1]L$ is a permutation of $\Alph(\mathbf{T}')$.  
Assume that $\Psi(\mathbf{T}')$ holds and each of $U$, $V$ contains all letters from $\Alph(\mathbf{T}')$.
%
%\[ |L|=|F|=|\alpha(L)|=|\alpha(F)|=k,\
% |U|=\gamma(k-1),|V|=\gamma(k-1+1.\]
Then $\D$ is an s-cover of $\mathbf{T}'$.
\end{claim}
\begin{proof} 
It suffices to show that each position in $U\cdot V$ is covered by $\D$.
Let $F=a_1a_2\cdots a_k,\, L=b_1b_2\cdots b_{k'}$ ($k' \in \{k-1,k\}$).
Let $i$ be a position in $\mathbf{T}'$ that is located in $UV$.
There are three critical cases when  $\mathbf{T}'[i]\in \{a_1,a_2\}$.
\cref{fig:3cases} illustrates the proof  for each of these cases.
\begin{description}
\item{\bf Case 1: $\mathbf{T}'[i]=a_1$.} 
Due to condition $\Psi(\mathbf{T}')$ we have $a_1=b_j$, $j<k'-1$.
Then $FL$ covers position $i$ as follows: $a_1a_2\cdots a_kb_1b_2\cdots b_{j-1}$ is a prefix of $F L^-$, $b_j=\mathbf{T}'[i]$, and $b_{j+1}\cdots b_{k'}$ is a suffix of $L$.
\item{\bf Case 2: 
$(\mathbf{T}'[i]=a_2)\wedge (i\in V)$.}

Again, due to $\Psi(\mathbf{T}')$, $a_2=b_j$, $j<k'$. Then $FL$ covers position $i$ as follows:
$a_1a_2\cdots a_kb_1b_2\cdots b_{j-2}$ is a prefix of $F L^-$, $b_{j-1}\in U$, $b_j=\mathbf{T}'[i]$, and $b_{j+1}\cdots b_{k'}$ is a suffix of $L$.
\item{\bf Case 3: $(\mathbf{T}'[i]=a_2)\wedge (i\in U)$.} 
Now we can select $a_3\in V$ and $\D$
covers $i$.
\end{description}

\noindent The condition $\Psi(\mathbf{T}')$ is crucial here.
The remaining cases, that is, $\mathbf{T}'[i]\notin \{a_1,a_2\}$, are straightforward (the condition $\Psi(\mathbf{T}')$ does not play a role). 

Indeed, let $\mathbf{T}'[i]=a_j$ with $j>2$.
Then $FL$ covers position $i$ as follows: $a_1a_2\cdots a_{j-1}$ is a prefix of $F$ (hence, of $F L^-$), $a_j=\mathbf{T}'[i]$, and $a_{j+1} \cdots a_k L$ is a suffix of $F_-L$ (as $j+1 >3$).
This completes the proof of the claim.
 \end{proof}

 By the claim, every position in the middle blocks of $\mathbf{T}\,=\, A\,B\,UV\,C\,D$ is covered by an occurrence of $FL$.
Now it remains to show how we can cover the remaining positions of $\mathbf{T}$.

\begin{description}
\item{\textbf{Case 1:} $i$ is to the left of $UV$.}
Let  $\mathbf{T}[i]=a_j$. Then, due to \cref{obs:March14}, there is a 
sequence of positions  $i_1<i_2<\cdots <i_j=i$ such that $\mathbf{T}[i_1]\mathbf{T}[i_2]\cdots \mathbf{T}[i_j]=a_1a_2\cdots a_j$.
The word $UVC$ has a subsequence $a_{j+1}a_{j+2}\cdots a_k$ and the word $D$ has a subsequence $L$.
Hence, $FL$ covers $i$.

\cref{notmiddle}  shows the case when $k=6$ and $\mathbf{T}_i=a_1$ or $\mathbf{T}_i=a_4$.  
\item{\textbf{Case 2:} $i$ is to the right of $UV$.} The proof is symmetric to Case 1.
\end{description}

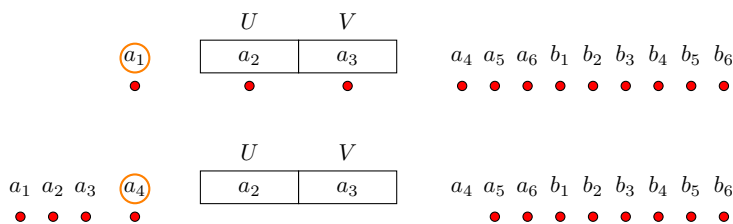
\begin{figure}[ht]
%\vspace*{-0.2cm}
\centering
    \begin{tikzpicture}[scale=0.87,transform shape]

\tikzstyle{dot}=[inner sep=0.05cm, circle, draw, fill=red]

\newcommand{\casesTemplate}[0]{
  %\foreach \l/\r/\label in {0/3/, 3/4.5/U, 4.5/6/V, 6/8.25/, 8.25/11.25/} {
  \foreach \l/\r/\label in {3/4.5/U, 4.5/6/V} {
    \draw (\l,0) rectangle (\r,0.5) node[midway,above=0.3cm] {$\label$};
  }
  %\node[above] at (2,0.55) {$i$};
  \node[above] (a2) at (3.75,0) {$a_2$};
  \node[above] (a3) at (5.25,0) {$a_3$};
  \foreach \i in {4,...,6} {
    \node[above] (a\i) at (5+\i*0.5,0) {$a_\i$};
  }
  \foreach \i in {1,...,6} {
    \node[above] (b\i) at (8+\i*0.5,0) {$b_\i$};
  }
}

\begin{scope}
  \casesTemplate
  \node[above] (a1) at (2,0) {$a_1$};
  \draw[orange,thick] (a1) circle (0.22cm);
  \foreach \p in {a1,a2,a3,a4,a5,a6,b1,b2,b3,b4,b5,b6} {
    \node[dot] at ($(\p.south)-(0,0.2)$) {};  
  }
\end{scope}

\begin{scope}[yshift=-2cm]
  \casesTemplate
  \foreach \i in {1,...,3} {
    \node[above] (aa\i) at (-0.25+\i*0.5,0) {$a_\i$};
  }
  \node[above] (aa4) at (2,0) {$a_4$};
  \draw[orange,thick] (aa4) circle (0.22cm);
  \foreach \p in {aa1,aa2,aa3,aa4,a5,a6,b1,b2,b3,b4,b5,b6} {
    \node[dot] at ($(\p.south)-(0,0.2)$) {};  
  }
\end{scope}

\end{tikzpicture}
    \caption{Illustration of covering a position $i\in AB$ (encircled; to the left of $UV$). We show  the cases $\mathbf{T}[i]=a_1$ and $\mathbf{T}[i]=a_4$.
    }\label{notmiddle}
\end{figure}

Assuming that $S$ is s-primitive, we derived a contradiction:  $FL$ is an s-cover of $\mathbf{T}$, 
and consequently, due to \cref{123},   of $S$. 
This completes the proof of the lemma.
\end{proof} 
\cref{April12_2024} and XY-Lemma (\cref{lem:comb}) directly imply the 
following upper bound
for general case.

\begin{theorem}[Upper bound]\label{thm: Improved upper bound}
For $k \ge 4$, $\gamma(k)\le (2k-2)\gamma(k-1)+6$.
\end{theorem}

Let us define \begin{equation}
\Delta(4)=19, \ \Delta(k)=(2k-2)\cdot \Delta(k-1)+6\ \mbox{for}\ k>4.
\end{equation}

In the conference version of the paper \cite{DBLP:conf/spire/Charalampopoulos22}, we showed \[\gamma(k)\le P(k),\ \mbox{where}\
P(k)=2k\cdot P(k-1)\mbox{ for }k>5, \ P(4)=19.\]
Here we improve this bound by replacing 
$P(k)$ by $\Delta(k)$.
We have \[\Delta(5)=158,\ \Delta(6)=1586,\ \Delta(7)=19038,\ \Delta(8)=266538.\]
\[P(5)=190,\ P(6)=2280,\ P(7)=32046,\ P(8)=512736.\]

Furthermore, it can be shown by induction that for $k \ge 4$, $\Delta(k) <2^{k-1}\,k!$. Therefore, \cref{{cor:Tomek 29.11}}, \cref{Kuba's mail} and \cref{thm: Improved upper bound} imply the following bounds.

\begin{corollary}\label{cor:asymp}
For $k \geq 4$, we have $$5\cdot 2^{k-1}-1 \le \gamma(k) < k!\cdot 2^{k-1}.$$
\end{corollary}

\begin{comment}
\subsection{Behaviour of the function $\gamma(k)$ for small $k$}\label{subsec:smallk}
\noindent The values of $\gamma$ for small $k$ are as follows (see also \cref{tab:gamma_ex}):

\begin{itemize}
\item $\gamma(1)=1$ -- trivial;

\item $\gamma(2)=3$ -- using square-free words;

\item $\gamma(3)=8$ -- due to \cref{fact} and \cref{fig:ternary};

\item $\gamma(4)=19$ -- through computer experiments.

\item 
$39 \le \gamma(5) \le 190$ -- due to Inequality~\eqref{eqstar} and $\gamma(4)=19$.
\end{itemize}

%\vspace*{-0.8cm}
\begin{table}[htbp]
    \centering
    \begin{tabular}{c|c|l}
         $k$ & $\gamma(k)$ & $\quad$examples of s-primitive words \\\hline\hline
         $\quad$1$\quad$ & 1 & $\quad a$ \\\hline
         2 & 3 & $\quad aba$ \\\hline
         3 & 8 & $\quad abcabacb$ \\\hline
         \multirow{2}{*}{4} & \multirow{2}{*}{$\quad$19$\quad$} & $\quad abacadbabdcabcbadac$ \\
         && $\quad abcdabacadbdcbabdac$ \\\hline
         $5$ & $\ge 39$ & $\quad abacadbabdcabcbadaceabacadbabdcabcbadac$
    \end{tabular}
    \vspace*{0.2cm}
    \caption{The values of $\gamma$ for small alphabet-size $k$.}
    \label{tab:gamma_ex}
\end{table}
\end{comment}

\section{Computing s-covers}\label{sec:slow_algo}
The following observation is a common property of s-covers and standard covers.

\begin{observation}\label{obs:simple}
If $C$ is an s-cover of $S$ and $C'$ is an s-cover of $C$, then $C'$ is an s-cover of $S$.
\end{observation}

\begin{theorem}\label{thm:algo}
Let $S$ be a length-$n$ word over an integer alphabet of size $k=n^{\cO(1)}$.
\begin{enumerate}[(a)]
    \item\label{ita} A shortest s-cover of $S$ can be computed in $\cO(n \cdot \min(2^n,(k-1)^{\gamma(k)}))$ time.
    \item\label{itb} One can check if $S$ is s-primitive and, if not, return a non-trivial s-cover of $S$ in $\cO(n + 2^{\gamma(k)} \gamma(k))$ time.
    \item\label{itc} An s-cover of $S$ of length at most $\gamma(k)$ can be computed in $\cO(n 2^{\gamma(k)} \gamma(k))$ time.
\end{enumerate}
\end{theorem}
\begin{proof}  We prove each of the three points separately.

\paragraph{\bf Point \eqref{ita}} By \cref{obs:simple}, a shortest s-cover of $S$ is s-primitive. By the definition of $\gamma(k)$, every shortest s-cover $C$ of $S$ is a word of length up to $\gamma(k)$. Moreover, every such $C$ starts with the same letter as $S$ and $C[i] \ne C[i-1]$ holds for all $i \in \{1,\ldots,|C|-1\}$. Thus the number of words that need to be tested is bounded by $\sum_{\ell=0}^{\gamma(k)-1} (k-1)^\ell = \cO((k-1)^{\gamma(k)})$. On the other hand, there are $2^n$ subsequences of $S$. Hence, there are $\min(2^n, (k-1)^{\gamma(k)})$ candidates to be checked. With the aid of the algorithm from \cref{lem:checking_alg} we can generate and check each candidate in $\cO(n)$ time. This gives the desired complexity.

\paragraph{\bf Point \eqref{itb}}
If $n \le \gamma(k)$, we can use the algorithm from \eqref{ita} which works in $\cO(2^{\gamma(k)}\gamma(k))$ time. Otherwise, $S$ is not s-primitive. We can find a non-trivial s-cover of $S$ as follows. 

Let $S=S'S''$ where $|S'|=\gamma(k)+1$. We can use the algorithm from \eqref{ita} to compute a shortest s-cover $C$ of $S'$ in $\cO(2^{\gamma(k)}\gamma(k))$ time. Then $C$ is a non-trivial s-cover of $S'$. Then, we can output $CS''$ as a non-trivial s-cover of $S$ (cf.\ \cref{obs:scover_factor}). This takes $\cO(n+2^{\gamma(k)}\gamma(k))$ time.

\paragraph{\bf Point \eqref{itc}}
By \cref{obs:simple}, any s-cover of an s-cover of $S$ will be an s-cover of $S$.
We can thus repeatedly apply the algorithm underlying \eqref{itb}, apart from outputting the computed non-trivial s-cover.
More precisely, in each step the current word is represented as a concatenation of a word of length up to $\gamma(k)$ and a suffix of $S$.

As each application of this algorithm removes at least one letter of $S$, the number of steps is at most $n-\gamma(k)$. Each step takes $\cO(2^{\gamma(k)}\gamma(k))$ time and hence the conclusion follows.
 \end{proof}

 \begin{corollary}
A shortest s-cover of a word over a constant-sized alphabet can be computed in linear time.
\end{corollary}

The algorithm from \cref{thm:algo}\eqref{ita} can be slightly improved to work in \[\cO(\min(2^nn,\,(n+\gamma(k))\alpha(k,2)^{\gamma(k)}))\]
time, where
\[\alpha(k,2)=k-1-\tfrac{1}{k-1}-\tfrac{1}{(k-1)^3}+\cO\left(\tfrac{1}{k^5}\right)\] is an upper bound on the growth rate of square-free words over an alphabet of size $k$; see \cite{DBLP:journals/mst/Shur14}.

\section{The number of distinct shortest s-covers}\label{sec:useless}
In the case of standard covers, if a word $S$ has two covers $C,C'$, then one of $C,C'$ is a cover of the other. This property implies, in particular, that a word has exactly one shortest cover.
In this section, we show that analogous properties do not hold for s-covers. There exist words $S$ having two s-covers $C,C'$ such that none of $C,C'$ is an s-cover of the other; e.g.\ $S=abcabcabcb$, $C=abcb$ and $C'=abcacb$. Moreover, a word can have many different shortest s-covers, as shown in \cref{thm:useless}.

\begin{figure}[htpb]
    \centering
\begin{tikzpicture}[xscale=0.3]
    \foreach \i/\c in {1/a,6/b,7/c,8/a,11/d,13/c,14/b,15/a}{
        \draw (\i,1) node[above] {$\textcolor{blue}{\c}$};
    }
    \foreach \i/\c in {1/a,2/b,3/c,4/a,5/d,9/c,10/b,12/a}{
        \draw (\i,0.5) node[above] {$\textcolor{blue}{\c}$};
    }
    \foreach \i/\c in {1/a,2/b,3/c,4/a,5/d,6/b,7/c,8/a,9/c,10/b,11/d,12/a,13/c,14/b,15/a}{

        \draw (\i,0) node[above] {$\textcolor{violet}{\c}$};
    }
\end{tikzpicture}
    \caption{$C_1=abca$ $d$ $cba$ is a shortest s-cover of $S=abca$ $d$ $bcacb$ $d$ $acba$ 
    (a palindrome). Hence $C_1^R=C_2=abc$ $d$ $acba$ is also an s-cover.}
    \label{fig:2-s-covers}
\end{figure}

We start with an example of a word $S$ of length 15 (see \cref{fig:2-s-covers})
with two different shortest s-covers and then extend it recursively.

\begin{lemma}\label{L15}
The word $S=abca\,d\,bcacb\,d\,acba$ has two different  s-covers of length $8$, $C_1=abca\,d\,cba$ and $C_2=abc\,d\,acba$ (cf.\ \cref{fig:2-s-covers}). It does not have any shorter  s-cover.
\end{lemma}

\begin{proof} Let $S\, =\, abcadbcacbdacba$.
\cref{fig:2-s-covers} illustrates that $S$ has two
s-covers of length $8$. We show they are the shortest s-covers.

Any s-cover of $S$ must contain the letter $d$ and before its first occurrence letters $a,b,c$ (in that order) must appear.
Symmetrically, after this letter, letters $c,b,a$ must appear.
The only word of length smaller than $8$ which satisfies this property is $abc$ $d$ $cba$; however, this is not an s-cover of $S$ (as it does not cover the middle letter $a$ in $S$).
 \end{proof}
 
\begin{remark} Assume $a_1,a_2,a_3,\cdots,a_m$ are new letters which are
pairwise distinct. 
    Let us take the word of the form $$ S\, a_1\,S\, a_2\, S\, a_3\,S\, \cdots\, S\, a_m,$$
    of length $n$. 
    
    This word
    has $2^{m}= 2^{\frac{n}{16}}$ distinct shortest s-covers, due to a straightforward extension of  \cref{obs:one_letter}.
    However,  the  number of distinct letters is here $\Omega(n)$.
    In the following theorem, we show that it can be lowered to $O(\log n)$.
    \end{remark}
    
\begin{theorem}\label{thm:useless}
For every positive integer $n$ there exists a word of length $n$ over an alphabet of size $\cO(\log n)$ that has at least $2^{\lfloor\frac{n+1}{16}\rfloor}$ different shortest s-covers.
\end{theorem}
\begin{proof}

Assume $a_1,a_2,a_3\cdots$ are new distinct letters. 
We  now construct a sequence of words $T_i$ such that 
$T_0=S\ \mbox{and}\ T_i=T_{i-1}a_iT_{i-1}\ \mbox{for}\ i > 0.$
The word $T_i$ has length $16\cdot 2^{i}-1=2^{i+4}-1$. Let us consider an infinite word $T$ containing each $T_i$ as its prefix  (it is well defined as each $T_i$ is a prefix of $T_{i+1}$). 
We show by induction, using \cref{obs:one_letter}, the following fact.

\smallskip
$T[0\dd n)$ has exactly $2^{\lfloor\frac{n+1}{16}\rfloor}$ different shortest covers.

\smallskip\noindent 
The base case for $n \le 15$ holds as every word has a shortest s-cover and for $n=15$ we apply \cref{L15} as $T[\dd 15)=S$.

Assume that $n > 15$ and $2^{i+4} \le n < 2^{i+5}$. Then 
\[T[0\dd n)=T_{i+1}[0 \dd n)\,=\,T_i\,a_{i+1}\,T_i[0 \dd n-2^{i+4}).\]
By \cref{obs:one_letter}, the number of shortest s-covers of $T[\dd n)$ is the number of shortest s-covers of $T_i$ times the number of shortest s-covers of $T[0 \dd n-2^{i+4})$, that is

\[2^{\frac{2^{i+4}}{16}} \cdot 2^{\lfloor\frac{n-2^{i+4}+1}{16}\rfloor}=2^{\lfloor\frac{n+1}{16}\rfloor},\]
as desired.
\end{proof}

 \section{Proof of XY-Lemma (Lemma~\ref{lem:comb})}\label{sec:matching_lemma}
 %although it can be easily 
 %done by a computer, testing all 4 (essential) possibilities for $X$ and %$4^6$ possibilities for $Y$.
  It will be convenient in the proof 
 to use  a "reversed" version of $\Phi$. 
% $$\Phi^R(X,Y)\; \stackrel{def}{\equiv}\; \Phi(X,Y^R).$$
We say that $X$ \emph{matches} $Y$ if $\Phi(X,Y^R)$, where $Y^R$ is 
the reverse of $Y$.
In other words, $X$ matches $Y$ if there are length-2 factors $a_1a_2$ and $b_1b_2$ of $X$ and $Y$, respectively, satisfying the inequalities represented by edges in \cref{fig:smile}.
Let $A,B,C,D$ be pairwise distinct letters.

\begin{figure}[h!]
\vspace*{0.3cm}
    \centering
    \begin{tikzpicture}[scale=0.67,transform shape]
\tikzstyle{elem}=[shape=rounded rectangle, inner xsep=0.4cm, inner ysep=0.15cm, draw, fill=brown!35!white];
\tikzset{font=\Large,color=blue};

\node[elem] (a) at (0,0) {};
\node[elem] (b) at (3,0) {};
\node[elem] (c) at (6,0) {};
\node[elem] (d) at (9,0) {};

\node [left] at (a.west) {$a_1$};
\node [right] at (b.east) {$a_2$};
\node [right] at (c.east) {$b_1$};
\node [right] at (d.east) {$b_2$};

\draw plot [smooth, tension=2] coordinates {(b.south) (4.5,-1) (c.south)};
\draw plot [smooth, tension=2] coordinates {(a.north) (3,1) (c.north)};
\draw plot [smooth, tension=2] coordinates {(a.north) (4.5,1.5) (d.north)};

\end{tikzpicture}
    \vspace*{0.2cm}
    \caption{$a_1a_2$ matches $b_1b_2$.
    The edges mean that two connected symbols are different.}
    \label{fig:smile}
\end{figure}
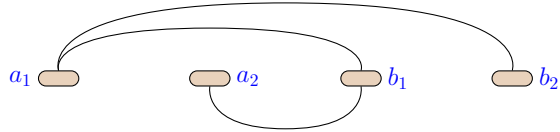
\begin{observation}\label{April30}~\\
\vspace*{-0.5cm}
\begin{enumerate}[(1)]
    \item\label{i1} If $Z_1,Z_2,Z$ are three different letters, then $Z_1Z$ matches $Z_2Z$.
    %$AB$ matches $CB$, $BC$ matches $AC$, $CA$ matches $BA$.\\
    %\item\label{i2} $AB$ matches $CD$, $BC$ matches $AD$, $CA$ matches $BD$.
    \item\label{i3} If $V\in \{AC,\,BA,\,CB,\, AD,\,BD,\,CD\}$, then $ABCA$ matches $V$.
\end{enumerate}
\end{observation}
\noindent  \cref{lem:comb} is readily equivalent to the  following fact:
\begin{claim}
 Assume $|X|=4$, $|Y|=6$ and both $X,Y$ are square-free.
 Then $X$ matches $Y$.
\end{claim}
 The proof is split into proofs of two 
claims. The first one concerns the 
situation
when $Y$ satisfies one of several simple conditions. 
The second one is when none of these conditions is satisfied.
Let $\Gamma=\{\lambda_1,\lambda_2,\lambda_3,\lambda_4\},$
\[\lambda_1=ABAC,\, \lambda_2=ABCA,\, \lambda_3=ABCB,\,  \lambda_4= ABCD.
\]
\begin{observation}
    Each $\lambda_i$ contains $AB$ and  contains $BA$ or $BC$.
\end{observation}
%
%\textcolor{red}{WR: zostawmy duże litery, lepiej widać o c chodzi. Wyjątek potwierdza regułę. WZ: Mozemy zmienic A,B,C na a,b,c zeby bylo bardziej spojne z reszta pracy - inicjalnie uzywalem wielkich liter na rysunku i tutaj w dowodach bo sa bardziej rozroznialne (np. male a i d sa u mnie dosyc podobne na rysunkach).}

%\medskip
\noindent
%It can be easily verified that
%\begin{claim}
    Each length-4 square-free word
    $X$ is equal, after one-to-one renaming of letters, to
    one of the words in $\Gamma$.
    The mapping of letters assigns, in a left-to-right order, to each letter of $X$ that does not have a previous occurrence the first unused letter in $A,B,C,D$.
    For  example, $ADAB$ is "isomorphic"
    to  $ABAC$.
%\end{claim}
%Let us first reduce the space of all the possible words $Y$. 
Let $\Gamma'$ be the set of all length-6 square-free 
words $Y$ satisfying at least one of the following 
conditions:
\begin{enumerate}
    \item the length-5 prefix of $Y$ contains a letter outside $\{A,B,C\}$
    \item  $Y$ contains $CB$
    \item the length-5 prefix of $Y$ contains $AC$
    \item the length-4 prefix of $Y$ contains $BA$.
\end{enumerate}
\begin{claim}\label{cl: cutting in the word}
If $X\in \Gamma$ and $Y\in \Gamma'$, then $X$
matches $Y$.
%and $X'[0] \not\in \{Y'[0],Y'[1]\}$, $Y'[0] \not\in \{X'[0],X'[1]\}$.
\end{claim}
\begin{proof}[Proof of the claim] We consider each condition separately.
We find matching length-2 factors $U$ and $V$ of $X$ and $Y$, respectively.

\paragraph{\bf Condition 1}
 If $Y[i] \not\in \{A,B,C\}$ for $i<5$, then let $V=Y[i\dd i+1]$. For each $X \in \Gamma$, we will identify a corresponding length-2 factor $U$ such that $U$ matches $V$. If $Y[i+1]=A$, we can take as $U$ any of the two words $BA$, $BC$, since each word in $\Gamma$ contains one of $BA$, $BC$ as a factor. If $Y[i+1]\neq A$, then take $U=AB$.

\paragraph{\bf Condition 2} If $Y$ contains $V=CB$, we take $U=AB$. Each 
word in $\Gamma$  contains $AB$.

\paragraph{\bf Condition 3} If $Y$ contains $V=AC$, then it matches 
$\lambda_2,\lambda_3,\lambda_4$ since each of them contains 
$U=BC$.

Hence we can assume that $X=\lambda_1=ABAC$.
Denote by $Z$ the next letter after $C$ in $Y$.
If $Z=B$ then $V=CB$ was already considered in the previous condition.
If $Z\notin \{A,B,C\}$, then we can take $U=AB$, $V=CZ$.
Hence we can assume $Z=A$ and take $V=CA$, which matches $BA$ which is contained in $\lambda_1$.

\paragraph{\bf Condition 4} If $BA$ appears in the length-4 prefix of $Y$, 
then we can assume that $BA$ is followed by $BC$ or $BZ$, where $Z\notin 
\{A,B,C\}$ (otherwise we would have one of the previous conditions or a square in $Y$). 
Then   $Y$ contains $BABC$ or $BABZ$. 
Each $\lambda_i$ matches $BABC$ and  $BABZ$. Hence $X$ matches $Y$.
%
%\smallskip This completes the proof of the claim.
 \end{proof}

%It is now sufficient to show the following claim.
\begin{claim}\label{cl: cutting in the word2}
If $X\in \Gamma$ and $Y\notin \Gamma'$ is a square-free word,
then $X$ matches $Y$.
%and $X'[0] \not\in \{Y'[0],Y'[1]\}$, $Y'[0] \not\in \{X'[0],X'[1]\}$.
\end{claim}

\begin{proof}%[{of the claim}]
The  
length-4 prefix $Y'$ of $Y$ contains only letters $A,B,C$
and contains none of  $BA,AC,CB,AA,BB,CC$ as its factor.
Hence 
%immediately after $A$ can only be $B$, after $B$ only $C$  and after $C$ %only $A$, consequently   
$$Y'\in \{ABCA,\,BCAB,\,CABC\}.$$
This implies that
each possible word $Y'$ contains  each of $AB,BC,CA$ as a factor.

Let us note that $AC$ in $\lambda_1$ matches $BC$ (see point \eqref{i1} of \cref{April30}), $CB$ in $\lambda_3$ matches $AB$, and $CD$ in $\lambda_4$ matches $AB$.
Hence, each $\lambda_i$, for  $i\ne 2$, matches one of $AB,BC,CA$. Consequently, if $i\ne 2$,
$X=\lambda_i$ matches $Y$ since $Y$ contains each of $AB,BC,CA$.
%\begin{itemize}
%\item If $X=\lambda_1=ABAC$ then $X$ is matched with CA, since $BA$ matches  %$CA$.\
%\item if $X=\lambda_3=ABCB$ then $X$ is matched with AB, since $CB$ matches %$AB$.
%\item If $X=\lambda_4=ABCD$ then  $X$  is matched with $AB$, since $CD$ %matches $AB$. 
%\end{itemize}
%

We are left with the case $X=\lambda_2=ABCA$. Observe that $Y'$   has period 3. We also know that $Y$ has length $6$ and is square-free.
Hence the period 3 of $Y'$ is broken at position $i=5$ or $i=6$. 
Let $V=Y[i-1]Y[i]$. 
We have that 
\[V\in \{AC,\,BA,\,CB\}\ \mbox{or}\ V\in \{A,B,C\}\cdot Z\ \text{where}\ Z\notin \{A,B,C\}.\]
%The condition $V\in \{AC,\,BA,\,CB\}$ is false, since $Y\notin \Gamma'$.
%Hence $V\in \{A,B,C\}\cdot D$.
 Then, due to point \eqref{i3} in 
\cref{April30}, $\lambda_2=ABCA$ matches $V$. Consequently, $X$
matches $Y$. This completes the proof of the claim.
\end{proof}
The thesis of \cref{lem:comb} follows from the three claims above.

\section{Conclusions and open problems}\label{sec:concl}
Subsequence covers are algorithmically and combinatorially more 
complicated than string covers. We solved several basic problems and left open several others. The main difficulty is the behaviour of the function 
$\gamma(k)$, which grows exponentially and in an irregular way.

\paragraph{Our combinatorial results}
These results are mainly related to the function $\gamma(k)$.

\begin{itemize}
\item Exact values of the function $\gamma$ for small arguments: $\gamma(1)=1$, $\gamma(2)=3$, $\gamma(3)=8$ (see \cref{fact}), $\gamma(4)=19$ (\cref{Kuba's mail}).
\item For $k \geq 4$, a recursive lower bound $\gamma(k) \ge 2\cdot\gamma(k-1)+1$ (\cref{thm:Tomek 29.11}) and a compact lower bound $\gamma(k) \ge 5\cdot 2^{k-2}-1$ (\cref{cor:Tomek 29.11}).
\item For $k \geq 4$, a recursive upper bound $\gamma(k)\le (2k-2)\gamma(k-1)+6$ (\cref{thm: Improved upper bound}) and a compact upper bound $\gamma(k) < k!\cdot 2^{k-1}$ (\cref{cor:asymp}).
\item For every $n>0$, a construction of a word of length $n$ over an alphabet of size $\cO(\log n)$ that has at least $2^{\lfloor\frac{n+1}{16}\rfloor}$ different shortest s-covers (\cref{thm:useless}).
\end{itemize}

%\{\textcolor{red}{\cref{tab:results}

\paragraph{Our algorithmic results}
\cref{tab:results} summarizes our algorithmic results, where $n$ denotes the length of the input string and $k$ the size of alphabet.

\begin{table}[h]
    \centering \footnotesize
    \begin{tabular}{c|c|c}
         \textbf{Problem} & \textbf{Complexity} & \textbf{Note} 
         \\[2mm]\hline\hline\rule{0pt}{4ex}  
         Testing if a single candidate is an s-cover & $\cO(n)$ & \cref{lem:checking_alg} \\[2mm]\hline\rule{0pt}{4ex} 
         Finding a shortest s-cover & $\cO(n\cdot\min(2^n,(k-1)^{\gamma(k)}))$ & \multirow{3}{*}{\rule{0pt}{8ex} \cref{thm:algo}} \\[2mm]\cline{1-2}\rule{0pt}{4ex} 
         Finding any s-cover & $\cO(n+2^{\gamma(k)}\gamma(k))$ & \\[2mm]\cline{1-2}\rule{0pt}{4ex} 
         Finding an s-cover of length at most $\gamma(k)$ & $\cO(n2^{\gamma(k)}\gamma(k))$ & \\
    \end{tabular}
    \caption{ %The table shows our algorithmic results, where $n$ denotes the length of the input string and $\gamma(k)$ denotes the length of a longest s-primitive word over an alphabet of size $k$.
    }
    \label{tab:results}
\end{table}

\paragraph{Open problems}
There are several natural algorithmic and combinatorial questions:
\begin{itemize}
    \item What is the complexity of checking if  a given word is s-primitive?
    We do not even know if it is in NP. We know it is in co-NP, because we can  guess  a witness (a non-trivial s-cover) and test it in linear time. 
    \item What is the complexity of computing the shortest s-cover of a  word?
    \item What is the complexity of computing  the number of different s-covers of a word?
    \item What is the exact value of $\gamma(5)$? We know only that $39\le \gamma(5)\le 158$.
    \item Let us define $\gamma'(1)=1,\ \gamma'(k+1)=2\,\gamma'(k)+k$ for $k>1$.
     We have
    $\gamma(k)=\gamma'(k)$ for $1\le k<5$. Is it always true?
    \item Is there a short and computer-free proof 
    that
     $abacadbabdcabcbadac$ is s-primitive?
\end{itemize}

We conjecture that the first three (algorithmic) problems are NP-hard for general alphabets.
The last three  (combinatorial) problems seem to be elusive.

\subsection*{Acknowledgements} We thank Juliusz Straszy\'nski for his help in conducting computer experiments.

\bibliographystyle{elsarticle-num}
\bibliography{references}

\begin{thebibliography}{10}
\expandafter\ifx\csname url\endcsname\relax
  \def\url#1{\texttt{#1}}\fi
\expandafter\ifx\csname urlprefix\endcsname\relax\def\urlprefix{URL }\fi
\expandafter\ifx\csname href\endcsname\relax
  \def\href#1#2{#2} \def\path#1{#1}\fi

\bibitem{DBLP:journals/ipl/ApostolicoFI91}
A.~Apostolico, M.~Farach, C.~S. Iliopoulos, Optimal superprimitivity testing
  for strings, Inf. Process. Lett. 39~(1) (1991) 17--20.
\newblock \href {https://doi.org/10.1016/0020-0190(91)90056-N}
  {\path{doi:10.1016/0020-0190(91)90056-N}}.

\bibitem{DBLP:journals/ipl/Breslauer92}
D.~Breslauer, An on-line string superprimitivity test, Inf. Process. Lett.
  44~(6) (1992) 345--347.
\newblock \href {https://doi.org/10.1016/0020-0190(92)90111-8}
  {\path{doi:10.1016/0020-0190(92)90111-8}}.

\bibitem{DBLP:journals/ipl/MooreS95}
D.~W.~G. Moore, W.~F. Smyth, A correction to "{A}n optimal algorithm to compute
  all the covers of a string", Inf. Process. Lett. 54~(2) (1995) 101--103.
\newblock \href {https://doi.org/10.1016/0020-0190(94)00235-Q}
  {\path{doi:10.1016/0020-0190(94)00235-Q}}.

\bibitem{DBLP:journals/tcs/CzajkaR21}
P.~Czajka, J.~Radoszewski, Experimental evaluation of algorithms for computing
  quasiperiods, Theor. Comput. Sci. 854 (2021) 17--29.
\newblock \href {https://doi.org/10.1016/j.tcs.2020.11.033}
  {\path{doi:10.1016/j.tcs.2020.11.033}}.

\bibitem{DBLP:journals/fuin/MhaskarS22}
N.~Mhaskar, W.~F. Smyth, String covering: {A} survey, Fundam. Informaticae
  190~(1) (2022) 17--45.
\newblock \href {https://doi.org/10.3233/FI-222164}
  {\path{doi:10.3233/FI-222164}}.

\bibitem{DBLP:journals/jcss/WarmuthH84}
M.~K. Warmuth, D.~Haussler, On the complexity of iterated shuffle, J. Comput.
  Syst. Sci. 28~(3) (1984) 345--358.
\newblock \href {https://doi.org/10.1016/0022-0000(84)90018-7}
  {\path{doi:10.1016/0022-0000(84)90018-7}}.

\bibitem{DBLP:conf/csr/RizziV13}
R.~Rizzi, S.~Vialette, On recognising words that are squares for the shuffle
  product, Theor. Comput. Sci. 956 (2023) 111156.
\newblock \href {https://doi.org/10.1016/J.TCS.2017.04.003}
  {\path{doi:10.1016/J.TCS.2017.04.003}}.

\bibitem{DBLP:journals/jcss/BussS14}
S.~Buss, M.~Soltys, Unshuffling a square is {NP}-hard, J. Comput. Syst. Sci.
  80~(4) (2014) 766--776.
\newblock \href {https://doi.org/10.1016/j.jcss.2013.11.002}
  {\path{doi:10.1016/j.jcss.2013.11.002}}.

\bibitem{lothaire_2002}
M.~Lothaire, Algebraic Combinatorics on Words, Encyclopedia of Mathematics and
  its Applications, Cambridge University Press, 2002.
\newblock \href {https://doi.org/10.1017/CBO9781107326019}
  {\path{doi:10.1017/CBO9781107326019}}.

\bibitem{DBLP:conf/cpm/KolpakovPPK14}
R.~Kolpakov, M.~Podolskiy, M.~Posypkin, N.~Khrapov, Searching of gapped repeats
  and subrepetitions in a word, J. Discrete Algorithms 46-47 (2017) 1--15.
\newblock \href {https://doi.org/10.1016/J.JDA.2017.10.004}
  {\path{doi:10.1016/J.JDA.2017.10.004}}.

\bibitem{DBLP:journals/tcs/BulteauV20}
L.~Bulteau, S.~Vialette, Recognizing binary shuffle squares is {NP}-hard,
  Theor. Comput. Sci. 806 (2020) 116--132.
\newblock \href {https://doi.org/10.1016/j.tcs.2019.01.012}
  {\path{doi:10.1016/j.tcs.2019.01.012}}.

\bibitem{DBLP:journals/tcs/FlouriIKPPST13}
T.~Flouri, C.~S. Iliopoulos, T.~Kociumaka, S.~P. Pissis, S.~J. Puglisi, W.~F.
  Smyth, W.~Tyczyński, Enhanced string covering, Theor. Comput. Sci. 506
  (2013) 102--114.
\newblock \href {https://doi.org/10.1016/j.tcs.2013.08.013}
  {\path{doi:10.1016/j.tcs.2013.08.013}}.

\bibitem{DBLP:journals/algorithmica/KociumakaPRRW15}
T.~Kociumaka, S.~P. Pissis, J.~Radoszewski, W.~Rytter, T.~Wale\'n, Fast
  algorithm for partial covers in words, Algorithmica 73~(1) (2015) 217--233.
\newblock \href {https://doi.org/10.1007/s00453-014-9915-3}
  {\path{doi:10.1007/s00453-014-9915-3}}.

\bibitem{DBLP:conf/esa/Radoszewski23}
J.~Radoszewski, Linear time construction of cover suffix tree and applications,
  in: 31st Annual European Symposium on Algorithms, {ESA} 2023, Vol. 274 of
  LIPIcs, Schloss Dagstuhl - Leibniz-Zentrum f{\"{u}}r Informatik, 2023, pp.
  89:1--89:17.
\newblock \href {https://doi.org/10.4230/LIPICS.ESA.2023.89}
  {\path{doi:10.4230/LIPICS.ESA.2023.89}}.

\bibitem{DBLP:conf/spire/Charalampopoulos22}
P.~Charalampopoulos, S.~P. Pissis, J.~Radoszewski, W.~Rytter, T.~Waleń,
  W.~Zuba, Subsequence covers of words, in: String Processing and Information
  Retrieval - 29th International Symposium, {SPIRE} 2022, Vol. 13617 of Lecture
  Notes in Computer Science, Springer, 2022, pp. 3--15.
\newblock \href {https://doi.org/10.1007/978-3-031-20643-6\_1}
  {\path{doi:10.1007/978-3-031-20643-6\_1}}.

\bibitem{DBLP:journals/mst/Shur14}
A.~M. Shur, Growth of power-free languages over large alphabets, Theory Comput.
  Syst. 54~(2) (2014) 224--243.
\newblock \href {https://doi.org/10.1007/S00224-013-9512-X}
  {\path{doi:10.1007/S00224-013-9512-X}}.

\end{thebibliography}

\end{document}